\documentclass[12pt]{amsart}
 
\usepackage{fullpage, amssymb, graphicx, url}
\def\dunce{\;\raisebox{-12mm}{\includegraphics[width=36mm]{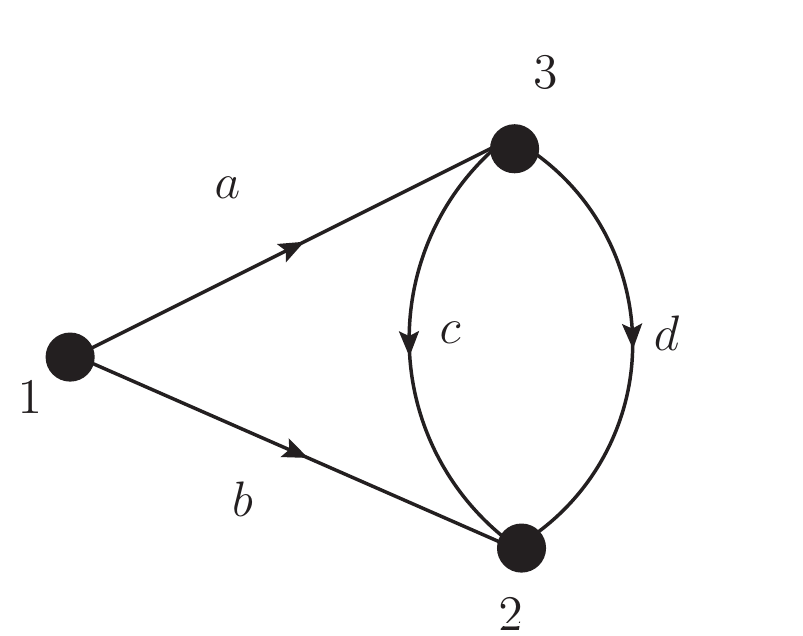}}\;}
\def\rk{\mathrm{rk}}

\newtheorem{thm}{Theorem}
\newtheorem{prop}[thm]{Proposition}
\theoremstyle{definition}
\newtheorem{definition}[thm]{Definition}

\title{Tensor structure from scalar Feynman matroids}
\author{Dirk Kreimer and Karen Yeats}
\thanks{Karen Yeats is supported by an NSERC Discovery grant, Dirk Kreimer
is supported by the Alexander von Humboldt Foundation
as an Alexander von Humboldt Professor,
endowed by the Federal Ministry of Education and Research, Germany.}

\begin{document}
\begin{abstract}
We show how to interpret the scalar Feynman integrals which appear when reducing tensor integrals as scalar Feynman integrals coming from certain nice matroids.
\end{abstract}
\maketitle

\section{Introduction}
Feynman integrals arise as evaluations of Feynman graphs by Feynman rules. 
In momentum space, these integrals $\int I_\Gamma$ are integrals over rational functions $I_\Gamma$
of internal loop momenta, 
with unit numerator in case of a scalar field theory. 

The tensor integrals appearing in the general case can be reduced to scalar integrals as well,
in particular, when we allow complex powers of propagators in the context of dimensional regularization
\cite{Tarasov}. Tensor integrals can then be reduced to integrals with unit numerator, on the expense 
of shifted dimensions and powers at propagators.

A basis of such scalar integrals which is sufficient to compute a desired amplitude to a given order
is often called a set of master integrals \cite{SmSm}. Integration by part (IBP) identities \cite{ChetTkach}
relate to tensor structure as in \cite{Laporta}, for example. The study of such master integrals, typically with arbitrary
complex powers of propagators, are the bread and butter of computational particle physics.
Such master integrals are not straightforwardly related to graphs though. 

Indeed, consider a Feynman integral arising from some graph, 
with a denominator given by a product of scalar propagators $P_{k_e}=1/(k_e^2+m_e^2)$ for each edge $e$.
With tensor structure, we will have scalar products $k_e\cdot k_j$ in the numerator.
Assuming that propagators $P_{k_e},P_{k_j},P_{k_f}$, $k_f=k_e-k_j$, appear in the denominator, we could resolve
\begin{equation}
 2 k_e\cdot k_j=P_{k_e}+P_{k_j}-P_{k_f}-m_e^2-m_j^2+m_{f}^2,
\end{equation}
eliminating the scalar product in favor of scalar integrals with possibly fewer propagators. 
The latter then correspond to 
graphs where an edge is contracted to a point.

But those propagators might not be present in the denominator. 
Then, the combinatorial interpretation in terms of graphs is missing.
Nevertheless, general products of propagators, in the denominator or numerator, 
even more generally with arbitrary complex powers, 
have proven useful in practical computations \cite{L09,L10}.
Indeed, any tensor integral can be reduced to a scalar integral on the expense of having a sufficiently general 
product of propagators at hand. 
But the combinatorial interpretation alluded to above in terms of graphs is not available when the product under consideration does not 
configure a graph.

If it does, we have the important tools of parametric representations via Kirchhoff polynomials available.
Being compatible with the raising of propagators to non-integer powers, 
these polynomials allow for systematic insights into the algebraic geometric properties of Feynman amplitudes 
and a satisfying mathematical understanding of these periods and functions 
\cite{Brbig,BrS,BrY,bek}.

Here, we answer the question what replaces such graph polynomials in the general case, 
when the product of propagators to start with does not configure a graph.
We will see that the notion of a graph is replaced by more general notion of matroid,
with the notion of one-particle irreducibility most crucially still being intact.
Hence, a systematic way to cope with tensor structure of Feynman graphs is to switch to what 
we might dub scalar Feynman matroids. 
\section{Matroids}

\subsection{Definitions}

There are many equivalent ways to define matroids.  The most useful for our purposes is the circuit definition.  A standard reference for matroid theory is \cite{ox}.

\begin{definition}
  A \emph{matroid} consists of a finite set $E$ and a set $\mathcal{C}$ of subsets of $E$ satisfying
  \begin{enumerate}
    \item $\emptyset \not\in \mathcal{C}$
    \item If $C_1, C_2 \in \mathcal{C}$ and $C_1\subseteq C_2$ then $C_1=C_2$.
    \item If $C_1, C_2 \in \mathcal{C}$, $C_1\neq C_2$ and $e\in C_1\cap C_2$ then there is a $C_3 \in \mathcal{C}$ with $C_3\subseteq (C_1\cup C_2)\smallsetminus e$.
  \end{enumerate}
\end{definition}

The set $E$ corresponds to the set of edges of a graph.  The elements of $\mathcal{C}$ are called the \emph{circuits} of the matroid and correspond to cycles (with no repeated vertices) of a graph.  Thus, in the case of graphs, the first two axioms state the obvious facts that the empty set is not a cycle and that no cycle contains a smaller cycle.  The third axiom is the interesting one and for graphs says that if two cycles share an edge, then together but without the edge they must form or contain a cycle.

Consequently, a graph defines a matroid, called the \emph{cycle matroid} of the graph, and any matroid which is the cycle matroid of some graph is called a \emph{graphic matroid}.

Another perspective is through the incidence matrix of a graph.  Take a graph $G$ with $e$ edges and $v$ vertices, and direct the edges.  The incidence matrix of $G$ is the $v\times e$ matrix $B=(b_{ij})$ with
\[
  b_{ij} = 
  \begin{cases} -1 & \text{ edge $j$ begins at vertex $i$} \\
    1 & \text{ edge $j$ ends at vertex $i$} \\
    0 & \text{ otherwise}
  \end{cases}
\]
For an example see equation \ref{Dirk matrix}.  To view this matrix as a matroid, let $E$ be the set of columns and let a circuit be a set of columns which is linearly dependent but with every proper subset linearly independent.  Any matrix can be viewed as a matroid in this way, and such matroids are called \emph{representable} (over the real numbers).  Note that all graphic matroids are representable.

Note also that elementary row operations on the matrix do not change the matroid, nor does removing any rows of all zeroes or scaling columns by nonzero scalars.  If we label the columns then we can also swap columns along with their labels without changing the matroid.  

One very nice property of matroids is that every matroid has a dual which generalizes the graph dual for planar graphs.  It will suffice for us to describe how to calculate the dual for representable matroids\footnote{see for instance \cite{ox} chapter 2, for the general definition and the proofs that the matroid dual generalized the planar dual of graphs and that the given representable calculation is well defined.}.  Take a representable matroid with matrix $M$.  Row reduce $M$ swapping columns and removing zero rows until it has the form
\[
  (I_n|D)
\]
where $I_n$ is the $n \times  n$ identity matrix and $D$ is $n\times m$.
Then the dual matroid is represented by the matrix 
\[
  (-D^T|I_m).
\]

Matroids which are the duals of graphic matroids are called \emph{cographic} matroids.  

The matroid of a graph cannot distinguish between a graph which is disconnected and one with the same components connected only at a vertex.  This is very natural for quantum field theory since the Feynman integral also can't make this distinction.  Despite this ambiguity, the notion of one-particle irreducible is well defined for matroids; a matroid $M = (E, \mathcal{C})$ is 1PI if it is \emph{bridgeless}, that is every $e\in E$ is in at least one circuit.

\subsection{Regular matroids and unique representability}

As noted above matrices which differ by elementary row operations, by nonzero column scalings and column swaps, and by rows of all zeroes represent the same matroid\footnote{We can also add to this list the action of any field automorphism, should we be working over a field which, unlike $\mathbb{R}$, has nontrivial automorphisms.}.  However, given a representable matroid $M$ there may be matrices which represent $M$ but which are \emph{inequivalent} using the above operations.  It is only in very nice cases that $M$ is \emph{uniquely representable}, that is, it has no inequivalent representations.  Fortunately some of these nice cases are central to our applications.

A matroid is \emph{regular} if it can be represented by a \emph{totally unimodular} matrix, that is a matrix for which every submatrix has determinant $-1$, $0$, or $1$.  Graphic matroids are examples of regular matroids since the incidence matrix is totally unimodular.  Cographic matroids are also regular.  Another characterization of regular matroids is that they are representable over any field, in fact any totally unimodular representation will do this.

Some results on unique representability.  
\begin{prop}\label{prop unique}
Let $I_r$ be the $r \times r$ identity matrix and let $\mathbb{F}_q$ be the finite field with $q$ elements.
\begin{itemize}
  \item Regular matroids are uniquely representable over every field.
  \item If $M$ is a matroid which is representable over $\mathbb{F}_2$ and over a field $F$ then $M$ is uniquely representable over $F$.
  \item Let $F$ be a field and let $(I_r|D_1)$ and $(I_r|D_2)$, matrices over $F$ with labelled columns, represent the same matroid with the same labelling.  Suppose that every entry of $D_1$ and $D_2$ is $0$ or $\pm 1$, then $(I_r|D_1)$ and $(I_r|D_2)$ are equivalent representations (Brylawski and Lucas 1976, see \cite{ox} Theorem 10.1.1).
  \item If $M$ is a matroid which is 3-connected
and representable over $\mathbb{F}_2$ and a field $F$ but not $\mathbb{F}_4$ then $M$ is uniquely representable over $F$ \cite{Wgf3}.
\end{itemize}
\end{prop}
The precise meaning of \emph{3-connectivity}, which appears above, is not crucial for the following\footnote{To be complete, matroid is 3-connected if it cannot be written as a direct sum or as a 2-sum of two nonempty matroids.  If $M_1 = (E_1, \mathcal{C}_1)$ and $M_2 = (E_2, \mathcal{C}_2)$ are matroids with $E_1\cap E_2 = \emptyset$ then the \emph{direct sum} of $M_1$ and $M_2$ is the matroid $(E_1\cup E_2, \mathcal{C}_1 \cup \mathcal{C}_2)$.  If $E_1\cap E_2 = \{p\}$ where $p$ is not a circuit of $M_1$ or of $M_2$ and appears in at least one circuit of each, then the \emph{2-sum} of $M_1$ and $M_2$ is the matroid with underlying set $E_1\cup E_2 \smallsetminus \{p\}$ and with circuits all the circuits of $M_1$ and $M_2$ which do not contain $p$ along with all sets of the form $C_1\cup C_2 \smallsetminus \{p\}$ where $C_i$ is a circuit of $M_i$ which contains $p$.}.  For a simple graph with at least 4 edges, matroid 3-connectivity corresponds to graph 3-connectivity.

The analogue of a spanning tree for a matroid is called a \emph{base} or \emph{basis}.  A base of $M=(E, \mathcal{C})$ is a maximal subset of $E$ which contains no cycle.  In terms of a matrix of a representable matroid a base is, as expected, a basis of the column space of the matrix.  One of the alternate characterizations of matroids is in terms of the set of bases instead of the set of circuits, which shows that the bases carry the full information of the matroid.

Since we have a notion analogous to the spanning tree, we can hope to form the Kirchhoff polynomial or first Symanzik polynomial.  However, there is one subtlety.  In view of the matrix-tree theorem we have two characterizations of the first Symanzik polynomial $\Psi_G$ of a connected graph $G$ with edge variable $a_e$ associated to edge $e$, namely
\begin{equation}\label{matrix tree}
   \sum_{\substack{T\text{ spanning}\\\text{tree of }G}}\prod_{e\not\in T}a_e
  = \Psi_G = \det \begin{pmatrix} \Lambda & \widetilde{B}^T \\ -\widetilde{B} & 0 \end{pmatrix}
\end{equation}
where $\Lambda$ is the diagonal matrix of the $a_e$ and $\widetilde{B}$ is the incidence matrix of $G$ with any one row removed.  Alternate ways to write the second equality include
\begin{equation}\label{alternate}
  \Psi_G = (\prod_{e\in G} a_e)\left(\det\widetilde{B} \Lambda \widetilde{B}^T|_{a_e \leftarrow \frac{1}{a_e}, e\in G}\right) = \det\widetilde{B^*} \Lambda \widetilde{B^*}^T
\end{equation}
where $\widetilde{B^*}$ is as $\widetilde{B}$ but for the dual of $G$.  

For matroids the first equality of \eqref{matrix tree} is always available.  However, the usefulness of the Kirchhoff polynomial for quantum field theory comes from the second equality of \eqref{matrix tree} which is one way to convert momentum space integrals to Schwinger parametric integrals.   Thus, unlike in \cite{BW}, we will look to the second equality (equivalently to \eqref{alternate}) in order to define the first Symanzik polynomial for the matroids of interest to us.

Note that the incidence matrix of a connected graph has one more row than its rank, and so, to translate to matroids, $\widetilde{B}$ should correspond to a representing matrix of full rank.  For uniquely representable matroids, we have a unique (up to unproblematic transformations) matrix to use for $\widetilde{B}$.  However the subdeterminants which, when squared, become the coefficients of the monomials in the determinant are no longer necessarily $\pm 1$, and so we will obtain a variant of \eqref{matrix tree}, 
\begin{equation}\label{weights}
\Psi_M = \sum_{T}w_T\prod_{e\not\in T}a_e
\end{equation}
 with positive weights $w_T$ on the terms, where the sum runs over the bases of a uniquely representable matroid $M$.  In the case of regular matroids we have a totally unimodular matrix, and so we retain exactly the identity \eqref{matrix tree}.

Likewise, adding in external edges to carry the external momenta we can form the second Symanzik polynomial by forming the same polynomial from this larger matrix and then taking the terms which are quadratic in the external momenta.

Let us pause here and compare incidence matrices with configuration polynomials, see  \cite{Pa,BlK1loop} where
the two common graph polynomials are treated as special cases of configuration polynomials.
This is not the space for a detailed discussion how graph polynomials appear as configuration polynomials, suffice it
to say that the starting point is the short exact sequence 
$$ 0\to \mathbb{Q}^r\to \mathbb{Q}^E\xrightarrow{\partial}\mathbb{Q}^{V,0}\to 0,$$
for a graph with $r$ loops, edges $E$, vertices $V$ and $\mathbb{Q}^{V,0}$ being the image of the boundary map 
$\partial$
\cite{BlK1loop}. 

Let us actually be very explicit and start with the Dunce's cap graph, where we label oriented edges 
$e\in \{a,b,c,d\}$ and vertices $v\in\{1,2,3\}$ as follows.
$$ \Gamma=\dunce $$
Note that we consider the case of a momentum $q$ entering at vertex 1, a zero-momentum at vertex 3, and $-q$
at vertex two.
We then have the incidence matrix (the superscripts indicate a consistent choice of momentum flow
at the corresponding edges, which are indeed represented by columns)
\begin{equation}\label{Dirk matrix}
B_\Gamma=\bordermatrix{ 
  & \text{\tiny{$-\ell$}}& \text{\tiny{$\ell+q$}} & 
  \text{\tiny{$-\ell+k$}} & \text{\tiny{$-k$}}\cr
& -1 & -1 & 0 & 0 \cr
&  0 & 1 & 1 & 1 \cr
&  1 & 0 & -1 & -1  
}.
\end{equation}
Let $C_a,C_b,C_c,C_d$ be the columns for the edges. There are three circuits, any two of them forming a basis for the first Betti homology. 
For example, $-C_a+C_b-C_c=0=-C_a+C_b-C_d$ determine two circuits as solutions of 
\[
 \sum_{e\in\{a,b,c,d\}}w_e C_e=0,
\]
 with coefficients in $\{-1,0,1\}$, as it befits a proper graph.

Following \cite{BlK1loop}, we then get the first graph polynomial as the determinant of the two by two matrix 
\[
\mathrm{det}\bordermatrix{ & & \cr
& a+b+c & a+b \cr
&  a+b & a+b+d 
}=(a+b)(c+d)+cd.
\]
Here, edge labels denote the variables, a chosen basis of circuits determines the diagonal entries, 
and the off-diagonal entries $(i,j)$ are determined by the edges common to circuits $C_i,C_j$.

Similarly, the second graph polynomial is also a configuration polynomial which can be obtained as the Pfaffian norm $N_{rp}$
or Moore determinant of the matrix
\cite{BlK1loop}
\[
N:=\bordermatrix{ & & \cr
& a+b+c & a+b & a\mu_a+b\mu_b+c\mu_c\cr
&  a+b & a+b+d & a\mu_a+b\mu_b+d\mu_d\cr
& a\bar{\mu_a}+b\bar{\mu_b}+c\bar{\mu_c} & a\bar{\mu_a}+b\bar{\mu_b}+d\bar{\mu_d} &
a\bar{\mu_a}\mu_a+b\bar{\mu_b}\mu_b+c\bar{\mu_c}\mu_c+d\bar{\mu_d}\mu_d
},
\]
with
\begin{eqnarray*}
 N_{rp}(N) & = & ((c + d)\underbrace{\overline{(\mu_a - \mu_b)}(\mu_a - \mu_b)}_{=q_1^2}  b\\ & &  
+ \underbrace{\overline{(\mu_a - \mu_c - \mu_d)}(\mu_a - \mu_c - \mu_d)}_{=q_3^2} cd)a + 
\underbrace{\overline{(\mu_b - \mu_c - \mu_d)}
(\mu_b - \mu_c - \mu_d)}_{=q_2^2} bcd,\nonumber
\end{eqnarray*}
writing four-momenta $q_i$ entering at vertices $i$ as quaternionic matrices \cite{BlK1loop}.
\section{Encoding tensor structure with matroids}

\subsection{From Feynman integral to matroid}

For mathematical physics, the value of a mathematical structure must be in what it can \emph{do}.  There are three main ways in which matroids are better than graphs for representing Feynman integrals.  

First, and most trivially, Feynman graphs are redundant in that non-isomorphic graphs can give the same Feynman integral.  For instance, graphs which differ by splitting at two vertices and rejoining with a twist have the same Feynman integral.  Matroids remove this redundancy exactly: graphs corresponding to the same integrand give the same matroid and vice versa, see \cite{BW}.

Second, if we allow matroids then every graph has a dual.  This should be useful wherever the planar dual is currently useful.  In particular it should be useful for the question of understanding the periods of high loop, massless, primitive $\phi^4$ graphs where the tools, \cite{BKphi4, Sphi4, Brbig, BrY}, can predominantly be translated over to matroids.

Third, and the topic of this paper, we can use matroids to give an easy and natural way to represent Feynman integrals with arbitrary numerator structure as scalar matroid integrals with appropriate powers of the propagators appearing. Consider as an example a Feynman graph with the following topology.
\[
  \includegraphics{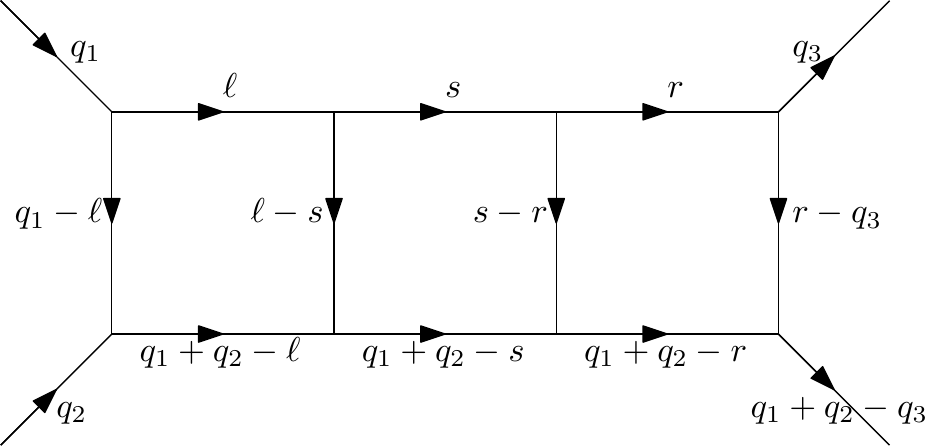}
\]
  The Feynman integral is of the form
\[
\int \frac{d^4\ell d^4s d^4r  }{\ell^2(\ell-s)^2s^2(s-r)^2r^2(q_1-\ell)^2(q_1+q_2-\ell)^2(q_1+q_2-s)^2(q_1+q_2-r)^2(r-q_3)^2}
\]
We understand that the propagators given might be raised to complex powers. This modifies the measure against which the graph polynomials are integrated, but does not modify the graph polynomial as a configuration polynomial.

If a similar topology is realized with quantum fields which do not sit in the trivial representation of the Poincar\'e group, we get tensor structure in the numerator as discussed in the 
introduction. For example, assume there is a $\ell\cdot r$ in the numerator.
  We are then  left with integrals of the form (general complex powers of propagators are still understood and they are typically shifted by integers when we resolve tensor structure in this manner)
\begin{equation}\label{big denom}
  \int \frac{d^4\ell d^4s d^4r }{\ell^2(\ell-s)^2s^2(s-r)^2r^2(q_1-\ell)^2(q_1+q_2-\ell)^2(q_1+q_2-s)^2(q_1+q_2-r)^2(r-q_3)^2(\ell-r)^2}
\end{equation}
The configuration of propagators given does not correspond to any graph.  
However, it does correspond to a matroid.  

We can see this by simply writing down a matrix which gives the desired matroid.  Each factor of the denominator of \eqref{big denom} is a generalized edge, that is an element of $E$, and thus a column in the matrix.  In the incidence matrix of a graph, the rows correspond to vertices.  If we modify the matrix by elementary row operations then the rows no longer correspond to vertices, but they remain sets of edges with momentum conservation.  For a matroid, vertices are no longer a well-defined concept, so we can only say that the rows are sets which conserve momentum.  With this in mind we can construct a matrix for the matroid $M$ of the denominator of \eqref{big denom}.
\begin{equation}\label{big matrix}
\bordermatrix{ 
  & \text{\tiny{$\ell$}}& \text{\tiny{$\ell-s$}} & 
  \text{\tiny{$s$}} & \text{\tiny{$s-r$}} & 
  \text{\tiny{$r$}} & 
  \text{\tiny{$q_1-\ell$}} & \text{\tiny{$q-\ell$}} & 
  \text{\tiny{$q-s$}} & \text{\tiny{$q-r$}} & 
  \text{\tiny{$r-q_3$}} & \text{\tiny{$q_1$}} & 
  \text{\tiny{$q_2$}} & 
  \text{\tiny{$q_3$}} & \text{\tiny{$q-q_3$}} & 
  \text{\tiny{$\ell-r$}} \cr
&  1 &-1 &-1 & 0 & 0 & 0 & 0 & 0 & 0 & 0 & 0 & 0 & 0 & 0 & 0 \cr
&  0 & 0 & 1 &-1 &-1 & 0 & 0 & 0 & 0 & 0 & 0 & 0 & 0 & 0 & 0 \cr
&  0 & 1 & 0 & 0 & 0 & 0 & 1 &-1 & 0 & 0 & 0 & 0 & 0 & 0 & 0 \cr
&  0 & 0 & 0 & 1 & 0 & 0 & 0 & 1 &-1 & 0 & 0 & 0 & 0 & 0 & 0 \cr
& -1 & 0 & 0 & 0 & 0 &-1 & 0 & 0 & 0 & 0 & 1 & 0 & 0 & 0 & 0 \cr
&  0 & 0 & 0 & 0 & 0 & 1 &-1 & 0 & 0 & 0 & 0 & 1 & 0 & 0 & 0 \cr
&  0 & 0 & 0 & 0 & 1 & 0 & 0 & 0 & 0 &-1 & 0 & 0 &-1 & 0 & 0 \cr
&  0 & 0 & 0 & 0 & 0 & 0 & 0 & 0 & 1 & 1 & 0 & 0 & 0 &-1 & 0 \cr
&  -1 & 0 & 0 & 0 & 1 & 0 & 0 & 0 & 0 & 0 & 0 & 0 & 0 & 0 & 1 
}
\end{equation}
where $q=q_1+q_2$.
The original graph corresponds to removing the final row and column.  This says that the original graph is $M$ with the new element corresponding to $(\ell-r)^2$ contracted.  The process of reversing a contraction in a matroid is called \emph{coextension}; it is not unique, but one element coextensions are completely characterized (see \cite{Bmat} pp156--158).

Further consider the circuits of $M$. Label the factors of \eqref{big denom} from left to right $a_1,\ldots, a_{11}$.  Take any cycle $C$ of the original graph.  Change variables in the internal momenta so that $C$ is indexed by one of the new internal momentum variables, $t$.  $C\cup \{a_{11}\}$ is a circuit of $M$ if $\ell-r$ depends on $t$ when written in the new variables, and $C$ itself is a circuit of $M$ otherwise.  This gives the circuits
\begin{gather}\label{raw cycles}
  \{a_1,a_2,a_6,a_7,a_{11}\},\{a_2,a_3,a_4,a_8\},\{a_4,a_5,a_9,a_{10},a_{11}\},\\ \{a_1,a_3,a_4,a_6,a_7,a_8,a_{11}\},\{a_2,a_3,a_5,a_8,a_9,a_{10},a_{11}\},\{a_1, a_3,a_5,a_6,a_7,a_8,a_9,a_{10}\}
\end{gather}
Comparing this to the original graph
\begin{gather*}
  \{a_1,a_2,a_6,a_7\},\{a_2,a_3,a_4,a_8\},\{a_4,a_5,a_9,a_{10}\},\\ \{a_1,a_3,a_4,a_6,a_7,a_8\},\{a_2,a_3,a_5,a_8,a_9,a_{10}\},\{a_1, a_3,a_5,a_6,a_7,a_8,a_9,a_{10}\}  
\end{gather*}
$a_{11}$ has simply been removed, which is another way to see that the original graph is $M$ with $a_{11}$ contracted.  

To include the information from the external momenta we can further put the external edges $e_1, e_2, e_3, e_4$ into $E$ and include as circuits (cycles through infinity) the momentum flow of the external momenta, and other possible external momentum flows coming from changes of variables.  $a_{11}$ will be included or not included in the circuits of $M$ by the same criterion as above.

This still does not give all the circuits of $M$.  Other circuits can come from pairs of cycles in the original graph which are joined in the coextension.  The following small graph example illustrates the situation.  Let
\[
  G = \raisebox{-1cm}{\includegraphics{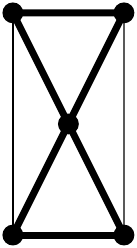}} \qquad H = \raisebox{-1cm}{\includegraphics{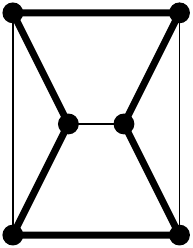}}
\]
Then $H$ is a coextension of $G$, and the pair of cycles in $G$ indicated by the fat lines become a single cycle in $H$.
Fortunately, given \eqref{raw cycles}, the axioms determine the remaining circuits, as if $C_1$ and $C_2$ are circuits of $M$ with $C_1\cap C_2 = \{a_{11}\}$.  Then $D = C_1\cup C_2 \smallsetminus \{a_{11}\}$, which is a pair of cycles in the original graph, must contain a circuit of $M$.  If either cycle of $D$ is itself a circuit of $M$ then $D$ cannot itself be a circuit of $M$ by the second axiom.  If neither cycle of $D$ is itself a circuit of $M$ then by the third axiom $D$ must be a circuit of $M$.  In the above case we add only the circuit
\[
  \{a_1,a_2,a_6,a_7,a_4,a_5,a_9,a_{10}\}
\]
Together this is all the circuits of $M$.

Note that we never used any graph specific properties of the original graph, and so we can iterate this procedure to remove all dot product factors from the numerator.  That we can keep the new matrices nice in the following sense
is our main result:

\begin{thm}\label{prop nice}
  Let $G$ be a graph and $P$ a set of pairs of edges of the graph (these are the pairs whose momenta appear dotted together in the numerator).  Applying the above construction with appropriate choices, we obtain a matrix row equivalent to
\begin{equation}\label{matrix form}
  \begin{pmatrix} I_{\rk G} & 0 & C \\ 0 & I_{r} & D \end{pmatrix}
\end{equation}
where $(I_{\rk G} | C)$ represents the matroid of $G$, $r\leq |P|$, and all entries of $C$ and $D$ are $0$ or $\pm 1$.
\end{thm}

\begin{proof}
  We need a few facts first.  If $A$ is a matrix with entries in $\{-1, 0, 1\}$ which represents a matroid $M$ both over $\mathbb{R}$ and over $\mathbb{F}_2$ then $A$ has no $2\times 2$ subdeterminant equal to $\pm 2$.  This is because otherwise $M$ would have a minor isomorphic to the matroid $U_{2,4}$, which is the matroid represented by 
\[
\begin{pmatrix} 1 & 0 & 1 & 1 \\ 0 & 1 & 1 & -1 \end{pmatrix}.
\]
But this is not possible for a matroid representable over $\mathbb{F}_2$. 

Next, if $A$ is a matrix with entries in $\{-1, 0, 1\}$ which represents a matroid $M$ both over $\mathbb{R}$ and over $\mathbb{F}_2$, and $R_1, R_2$ are two rows of $A$ which both have a nonzero entry in column $i$ then the linear combination of $R_1$ and $R_2$ with zero entry in column $i$ has all entries $0$ or $\pm 1$, and replacing $R_1$ by this new row still represents $M$ over both $\mathbb{R}$ and $\mathbb{F}_2$.

To see this, suppose otherwise.  Write $R_i = (r_{i,1}, \ldots , r_{i,m})$.  The only possible other entry is $\pm 2$.  If $\pm 2$ appears in column $j$ then $r_{1,j}r_{2,i}-r_{2,j}r_{1,i}  = \pm 2$ and the $2 \times 2$ matrix made of columns $i$ and $j$ of $R_1$ and $R_2$ has determinant $\pm 2$ which is impossible.  Replacing $R_1$ by this new row only involves scaling by $\pm 1$ and so this new matrix still represents $M$ both over $\mathbb{R}$ and over $\mathbb{F}_2$.  

Now returning to the original problem, row reduce the incidence matrix of $G$ using the above tools and remove any zero rows to obtain a matrix of the form $(I_{\rk G} | C)$ with the entries of $C$ in $\{0, \pm 1\}$ and which represents the matroid of $G$ over $\mathbb{R}$ and $\mathbb{F}_2$.

We inductively coextend for each element of $P$.  The base case is above.  Suppose we have a matrix of the form 
\[
 A =  \begin{pmatrix} I_{\rk G} & 0 & C \\ 0 & I_{r} & D \end{pmatrix}
\]
and a pair of edges $(e_1, e_2)$ of $G$.  If any edge of $G$ or new edge which we have already added from $P$ is already a linear combination with both weights nonzero of the momenta of $e_1$ and $e_2$ then we don't need a new edge and so just discard this pair.

Otherwise, $e_1$ and $e_2$ correspond to two column indices $i$ and $j$ which are either in the $I_{\rk G}$ part or the $C$ part of $A$.  If $i$ and $j$ are both in the $C$ part then coextend $A$ to
\[
\begin{pmatrix} I_{\rk G} & 0 & 0 & C \\ 0 & I_{r} & 0 & D  \\ 0 & 0 & 1 & v\end{pmatrix}
\]
where $v$ has either $\pm 1$ in positions $i$ and $j$ and zero elsewhere.  

If $i$ is in the $I_{\rk G}$ part and $j$ is in the $C$ part then let $\epsilon$ be the $i,j$ entry of $A$.  Build the new matrix
\[
A' = \begin{pmatrix} I_{\rk G} & 0 & 0 & C \\ 0 & I_{r} & 0 & D  \\ v_1 & 0 & 1 & v_2\end{pmatrix}
\]
where $v_1$ has $-1$ in position $i$ and $0$ elsewhere and $v_2$ has 
  $\begin{cases} \pm 1 & \text{if } \epsilon = 0 \\ -\epsilon & \text{otherwise}
  \end{cases}$ in position $j$ and $0$ elsewhere.
Then we can row reduce $A'$ by adding row $i$ to the last row and obtain a matrix in the form of \eqref{matrix form}.

Finally if both $i$ and $j$ are in the $I_{\rk G}$ part then we will make a matrix of the form
\[
A' = \begin{pmatrix} I_{\rk G} & 0 & 0 & C \\ 0 & I_{r} & 0 & D  \\ v & 0 & 1 & 0\end{pmatrix}
\]
with $v$ nonzero only in entries $i$ and $j$.  Call these entries $v(i)$ and $v(j)$. Let $R_i$ and $R_j$ be rows $i$ and $j$ of $C$.  If there is no column for which $R_i$ and $R_j$ are both nonzero then any choice of $\pm 1$ is fine for $v(i)$ and $v(j)$.  If $R_i$ and $R_j$ are both nonzero in column $k$ then let $v(i)=r_{i,k}$ and $v(j) = -r_{j,k}$.  Then row reducing $A'$ to clear entries $v(i)$ and $v(j)$ puts the row $-r_{i,k}R_i+r_{j,k}R_j$ below $D$, but by the facts we initially observed this also has all entries $0$ or $\pm 1$ and so we again obtain a matrix in the required form.
\end{proof}

Consider the big example again.  In the big example we have only added one non-graph edge, so the situation is not too complicated.  If we swap columns and row reduce the original matrix so that the original columns 1 and 5 are not part of the identity portion of the row reduced matrix, then we fall into the first case of the theorem when adding our new edge.  Specifically our new matrix will have the form
\[
  \begin{pmatrix} I_{\rk G} & 0 & C \\ 0 & 1 & v \end{pmatrix} 
\]
If we made other choices in row reduction we could end up in other cases.  When adding subsequent additional edges we can be forced into the other cases.

\subsection{From matroid to Feynman integral}

We can go from a matrix like \eqref{big matrix}, or any matrix of the form \eqref{matrix form}, (along with the information of which edges are external) back to a scalar Feynman integral by assigning a momentum variable for each edge and then applying Feynman rules in the sense that each internal edge contributes a factor $1/k^2$ 
raised to some complex power and each row of the matrix contributes a delta function corresponding to reading across the row.  The rows generate the set of momentum preserving subsets of edges, and so the counting works exactly as in the graph case even though we have no notion of vertices.  Alternately, we can build parametric Feynman integrals from the first and second Symanzik polynomials in the usual way.

In order to convert a real-representable matroid back into a Feynman integral, first notice that row operations do not change the Feynman integral since matrices which differ only by row operations simply lead to different but equivalent products of delta functions.  Rows of all zeros simply contribute nothing at all.  Column swaps and column scalings simply rename or scale the internal momenta and so do not change the Feynman integral or scale the entire integral.  Thus equivalent representations lead to essentially the same Feynman integral.   

By Theorem \ref{prop nice} the matroids we are interested always have a representation in the form $(I|D)$ with $D$ having entries $0$ or $\pm 1$; by the equivalent representation result of Brylawski and Lucas (see Proposition \ref{prop unique}) such representations are equivalent given that we have fixed labels for the edges.  So even if our matroids are not uniquely representable they have favoured representations which are all equivalent.  Furthermore for matroids in this form, only column scalings by $\pm 1$ will arise, so we do not have the scale ambiguity mentioned in the previous paragraph.  So we can build our Feynman integrals in either momentum space or parametric space out of these representations.  Thus we can speak of Feynman integrals associated to such matroids, and we can decompose tensor integrals into scalar integrals of matroids.

In the big example the matroid, which is represented by \eqref{big matrix}, is cographic, so in particular it is regular and so uniquely representable.

Finally, note that weights in the base expansion of the first Symanzik polynomial, as in \eqref{weights}, really do come up in physically relevant cases, and hence the matroid is not always regular, and so in particular is not always cographic.  Consider the complete bipartite graph $K_{3,3}$.  It is represented by the matrix
\[
\begin{pmatrix} 
1 & 0 &  0 & 1 & 1 & 0 & 0 & 0 & 0 \\
0 & 1 &  0 & 0 & 0 & 1 & 1 & 0 & 0 \\
0 & 0 & 1 & 0&  0&  0&  0&  1&  1 \\
0 & 0 & 0 &-1&  0& -1&  0& -1&  0 \\
0 & 0 & 0 & 0& -1&  0& -1&  0& -1 
\end{pmatrix}
\]
Row reduced
\[
\begin{pmatrix} 
1 & 0 &  0 & 0 & 0 & -1 & -1 & -1 & -1 \\
0 & 1 &  0 & 0 & 0 & 1 & 1 & 0 & 0 \\
0 & 0 & 1 & 0&  0&  0&  0&  1&  1 \\
0 & 0 & 0 & 1&  0&  1&  0&  1&  0 \\
0 & 0 & 0 & 0&  1&  0&  1&  0&  1 
\end{pmatrix}
\]
Pick edges 7 and 8, which do not share a vertex.  There are two possibilities for the coextension which are consistent with the preceding discussion. 
\[
\begin{pmatrix} 
1 & 0 &  0 & 0 & 0 & 0 & -1 & -1 & -1 & -1 \\
0 & 1 &  0 & 0 & 0 & 0 & 1 & 1 & 0 & 0 \\
0 & 0 & 1 & 0&  0&  0 & 0&  0&  1&  1 \\
0 & 0 & 0 & 1&  0&  0 & 1&  0&  1&  0 \\
0 & 0 & 0 & 0&  1&  0 & 0&  1&  0&  1 \\
0 & 0 & 0 & 0 & 0 & 1 & 0 & 1 & 1 & 0
\end{pmatrix}
\quad \text{and} \quad
\begin{pmatrix} 
1 & 0 &  0 & 0 & 0 & 0 & -1 & -1 & -1 & -1 \\
0 & 1 &  0 & 0 & 0 & 0 & 1 & 1 & 0 & 0 \\
0 & 0 & 1 & 0&  0&  0 & 0&  0&  1&  1 \\
0 & 0 & 0 & 1&  0&  0 & 1&  0&  1&  0 \\
0 & 0 & 0 & 0&  1&  0 & 0&  1&  0&  1 \\
0 & 0 & 0 & 0 & 0 & 1 & 0 & 1 & -1 & 0
\end{pmatrix}
\]
In the first case the submatrix with all rows and with columns 2,5,6,8,9,10 has determinant $2$ while in the second case the submatrix with all rows and columns 2,3,4,5,8,9 has determinant $2$.  So in either case a coefficient of $4$ will appear in the first Symanzik polynomial which can be verified by direct computation.   Both of these matroids are thus not regular.  Computing with Macek \cite{macek} we can check that these are representable over $\mathbb{F}_3$ and $\mathbb{R}$ but not $\mathbb{F}_4$ and hence by results of Whittle (see Proposition \ref{prop unique}) are none-the-less uniquely representable.

\section{Discussion}

Generalizing Feynman graphs to matroids gives us combinatorial interpretations for scalar master integrals, without sacrificing the graph based tools and definitions we need, such as 1PI, duality, contraction, and deletion.  In fact it improves these tools in that duality becomes defined for all graphs.

Moving to matroids also suggests a hierarchy of difficulty.  Planar graphs are most straightforward; both they and their duals are graphs.  General graphs and cographs come next; cographs are less familiar but they behave very much as graphs do.  Cographs are a natural next term for any series which begins with a planar piece and doesn't continue with graphs themselves.
Very slightly more than this is the class of regular matroids.  Regular matroids are nice in many ways, notably they are uniquely representable over every field and they have a matrix-tree theorem identical to that of graphs; typically one's graph based intuition is valid.  Finally, the most general matroids we need are more subtle, but do always have a nice representation in the form $(I|D)$ with $D$ having entries $0, 1, -1$.

\bibliographystyle{plain}
\bibliography{main}

\begin{thebibliography}{10}

\bibitem{bek}
Spencer Bloch, H\'el\`ene Esnault, and Dirk Kreimer.
\newblock On motives associated to graph polynomials.
\newblock {\em Commun. Math. Phys.}, 267:181--225, 2006.
\newblock arXiv:math/0510011v1 [math.AG].

\bibitem{BlK1loop}
Spencer Bloch and Dirk Kreimer.
\newblock Feynman amplitudes and {L}andau singularities for 1-loop graphs.
\newblock arXiv:1007.0338.

\bibitem{BW}
Christian Bogner and Stefan Weinzierl.
\newblock Feynman graph polynomials.
\newblock arXiv:1002.3458.

\bibitem{BKphi4}
D.J. Broadhurst and D.~Kreimer.
\newblock Knots and numbers in $\phi^4$ theory to 7 loops and beyond.
\newblock {\em Int.J.Mod.Phys.}, C6(519-524), 1995.
\newblock arXiv:hep-ph/9504352.

\bibitem{Brbig}
Francis Brown.
\newblock On the periods of some {F}eynman integrals.
\newblock arXiv:0910.0114.

\bibitem{BrS}
Francis Brown and Oliver Schnetz.
\newblock A {K3} in $\phi^4$.
\newblock arXiv:1006.4064.

\bibitem{BrY}
Francis Brown and Karen Yeats.
\newblock Spanning forest polynomials and the transcendental weight of
  {F}eynman graphs.
\newblock arXiv:0910.5429.

\bibitem{Bmat}
Thomas Brylawski.
\newblock {\em Theory of Matroids}, chapter 7: Constructions.
\newblock Encyclopedia of Mathematics and its applications. Cambridge
  University Press, 1986.

\bibitem{ChetTkach}
K.~G. Chetyrkin and F.~V. Tkachov.
\newblock Integration by parts: The algorithm to calculate beta functions in 4
  loops.
\newblock {\em Nucl. Phys. B}, 192:159--204, 1981.

\bibitem{macek}
Petr Hlin\v{e}n\'{y}.
\newblock The \textsc{Macek} program.
\newblock \url{http://www.mcs.vuw.ac.nz/research/macek}, 2007.
\newblock version 1.2.11.

\bibitem{Laporta}
S.~Laporta.
\newblock High-precision calculation of multi-loop {F}eynman integrals by
  difference equations.
\newblock {\em Int.\ J.\ Mod.\ Phys.\ A}, 15:5087--5159, 2000.
\newblock arXiv:hep-ph/0102033.

\bibitem{L10}
R.~N. Lee.
\newblock Calculating multiloop integrals using dimensional recurrence relation
  and {D}-analyticity.
\newblock arXiv:1007.2256.

\bibitem{L09}
R.~N. Lee.
\newblock Space-time dimensionality {D} as complex variable: calculating loop
  integrals using dimensional recurrence relation and analytical properties
  with respect to {D}.
\newblock {\em Nucl.~Phys.~B}, 830:474--492, 2010.
\newblock arXiv:0911.0252.

\bibitem{ox}
James Oxley.
\newblock {\em Matroid Theory}.
\newblock Oxford, 1992.

\bibitem{Pa}
Eric Patterson.
\newblock On the singular structure of graph hypersurfaces.
\newblock arXiv:1004.5166.

\bibitem{Sphi4}
Oliver Schnetz.
\newblock Quantum periods: A census of $\phi^4$-transcendentals.
\newblock arXiv:0801.2856.

\bibitem{SmSm}
A.~V. Smirnov and V.~A. Smirnov.
\newblock On the reduction of {F}eynman integrals to master integrals.
\newblock In {\em Proceedings of ACAT}, page~85, 2007.
\newblock arXiv:0707.3993.

\bibitem{Tarasov}
O.~V. Tarasov.
\newblock Connection between {F}eynman integrals having different values of the
  space-time dimension.
\newblock {\em Phys. Rev. D}, 54:6479--6490, 1996.
\newblock arXiv:hep-th/9606018.

\bibitem{Wgf3}
Geoff Whittle.
\newblock On matroids representable over {$GF(3)$} and other fields.
\newblock {\em Trans. AMS.}, 349(2):579--603, 1997.

\end{thebibliography}

\end{document}